\documentclass[12pt]{article}
\usepackage{graphicx} 
\usepackage{setspace} 
\usepackage[margin=1in]{geometry} 
\usepackage{subcaption}
\usepackage{balance}
\usepackage[linesnumbered,ruled,vlined]{algorithm2e}
\SetKwInOut{Parameter}{Parameters}
\SetKwInOut{Output}{Output}
\usepackage{multicol}  
\usepackage{amsmath}
\usepackage{amsthm}
\usepackage{MnSymbol}
\usepackage{authblk}

\usepackage{amsthm}
\newtheorem{theorem}{Theorem}[section]
\theoremstyle{definition}

\newtheorem{lemma}[theorem]{Lemma}
\newtheorem{definition}[theorem]{Definition}

\newtheorem{claim}[theorem]{Claim}

\usepackage{mathtools}
\DeclarePairedDelimiter\floor{\lfloor}{\rfloor}
\DeclarePairedDelimiter{\ceil}{\lceil}{\rceil}

\usepackage{hyperref}
\usepackage{xcolor}
\hypersetup{
    colorlinks,
    linkcolor={red!50!black},
    citecolor={blue!50!black},
    urlcolor={blue!80!black}
}

\usepackage{soul}

\def\inp{v_{in}}
\def\outp{v_{out}}
\def\binset{\{0,1\} }
\def\GBA{\mbox{\tt GBA}}
\def\RCA{\mbox{\tt RCA}}
\def\cA{\mathcal{A}}
\def\GradedBA{\mbox{\tt Graded\_BA}}

\title{Recursive Energy Efficient Agreement}
\author[1]{Shachar Meir}
\author[1]{ David Peleg}
\affil[1]{Weizmann Institute of Science, Israel}
\date{\today}

\begin{document}
\maketitle

\begin{abstract}
    Agreement is a foundational problem in distributed computing that have been studied extensively for over four decades. Recently, Meir, Mirault, Peleg and Robinson introduced the notion of \emph{Energy Efficient Agreement}, where the goal is to solve Agreement while minimizing the number of round a party participates in, thereby reducing the energy cost per participant. We show a recursive Agreement algorithm that  has $O(\log f)$ active rounds per participant, where $f<n$ represents the maximum number of crash faults in the system.
\end{abstract}

\section{Introduction}

\subsection{Motivation}


Agreement \cite{lamport1982byzantine} is a classical and foundational problem that has been studied extensively for over four decades. Informally, the problem requires that $n$ processors agree on a single value, even when up to $f$ of them are faulty.
Since Agreement forms a basic building block for many real-world system (e.g, distributed data bases and blockchain systems), significant effort was devoted to optimizing some cost measures of agreement algorithms (e.g., communication or time), and introducing additional desired properties (e.g., early stopping \cite{early-stopping}) that improve performance in practical systems.

Recently, Meir, Mirault, Peleg and Robinson \cite{MMPR25} introduced and studied the notion of ``energy efficient" agreement, where the goal is to minimize the energy usage of every processor, making agreement protocols more economical and environmentally friendly. 
In this paper, we continue that work 
by considering Agreement in the \emph{sleeping} model \cite{chatterjee2020sleeping} under two failure models: crash faults (where processors stop functioning permanently) and Byzantine faults (where processors deviate arbitrarily from the algorithm).

The (synchronous) sleeping model extends the standard (synchronous) messages passing model. In any round $t$, each processor in the network can voluntarily ``go to sleep" for any number of rounds $r$ of its choice. That is, a processor can decide in round $t$ to sleep in rounds $t+1, t+2, \ldots, t+r$ and wake up in round $t+r+1$. 
While a processor $p$ is sleeping, it cannot send or receive messages. 
Moreover, every message that is sent to $p$ in rounds $[t+1, t+r]$ is permanently lost, and never reaches its destination (not even in round $t+r+1$, when $p$ wakes up again). The main complexity measure in the sleeping model, in addition to rounds and message complexities, is \emph{awake} complexity: the maximum number of rounds any single processor spends awake during any execution of the algorithm.

\subsection{Known results}

In \cite{MMPR25}, Meir, Mirault, Peleg and Robinson showed two Crash Agreement (CA) algorithms that tolerate any number of crash faults $f<n$ and have optimal round complexity $f+1$. The first algorithm, designed for multi-valued Agreement (i.e, the input values are in some unknown set $V$ of values), has awake complexity $O(\floor{f^2/n})$. The second algorithm works only for binary CA (i.e, the input values are from $\binset$), but improves the awake complexity to $O(\floor{f/\sqrt{n}})$.

In a separate line of work, Momose and Ling \cite{momose_ling} proposed a recursive \emph{authenticated}\footnote{An authenticated algorithm is one that assumes and uses a digital signature infrastructure}
Byzantine Agreement (BA) algorithm that operates under the Byzantine failure model, for $f<n/2$, and achieves quadratic message complexity and $O(n)$ round complexity. 
Their approach works by recursively partitioning the set of processors. At each recursion level, the processors are split into subgroups ($Q_{2w}, Q_{2w+1}$), and agreement is reached via a composition of recursive calls and invocations to a black box procedure named $\GradedBA$ (whose definition is presented in Section \ref{sec:gba}).
More precisely, the fault tolerance threshold of the recursive algorithm can be expressed as $f < \min\{1/2, \beta\}\cdot n $, where $\beta$ is the fault tolerance threshold of the 
implementation used
for the $\GradedBA$ black box.

A useful observation is that the recursive algorithm itself does not use digital signatures explicitly (and thus it is an unauthenticated algorithm on its own), 
but the specific implementation of the $\GradedBA$ black box presented in~\cite{momose_ling}, which achieves resilience for failure thresholds up to $\beta = 1/2$, does require digital signatures. 
It follows that using an unauthenticated implementation of the $\GradedBA$ black box will render the recursive algorithm unauthenticated as well.

While \cite{momose_ling} does not consider energy complexity explicitly, it is easy to verify that the proposed recursive algorithm 
has awake complexity $O(\log n)$. This is because each processor is part of a recursive step  execution ($O(1)$ rounds locally) only when it is part of the subset of processors running the step, and the recursive execution creates a complete binary tree of subsets, which implies that each processor belongs to at most $O(\log n)$ subsets. 

\subsection{Our contributions}

In this work, we show that the approach of~\cite{momose_ling}
can be adapted to the \emph{crash} fault setting to yield a CA algorithm that tolerates any number of $f<n$ crash faults  and has awake complexity $O(\log f)$ and round complexity $O(f)$. 
Contrasted with the algorithms of \cite{MMPR25},
the following dichotomy emerges.
For a low number of faults, $f \leq \sqrt{n}\cdot g(n)$, where $g(n) = o(\log n)$, the algorithms of \cite{MMPR25}
have lower awake complexity.
Conversely, in the opposite case, the algorithm presented herein 
performs better.
%
This dichotomy suggests that awake complexity does not necessarily behave in a straightforward manner, raising interesting questions regarding potential hybrid algorithms and lower bounds on energy efficient crash agreement. 
%
%

We also show an \emph{unauthenticated} $\GradedBA$ implementation that works correctly when $f<n/3$, and that, when combined with the algorithm of~\cite{momose_ling}
and the optimization technique in Section \ref{sec:optimized_rec_alg}, yields an energy efficient \emph{unauthenticated} BA algorithm with awake complexity $O(\log f)$ and round complexity $O(f)$ that works correctly when $f<n/3$. 
%
%

\section{Preliminaries}
We assume a set $Q$ of $n$ processors, each with a unique ID in $[1,n]$. We assume a synchronous clique network, i.e, every two processor have a communication link, and time is partitioned into rounds of communication. Every round has the following structure. (i) Every processor sends some (or no) messages (ii) Every message that was sent at the beginning of the round is received by its destination (iii) Every processor performs local computation.


We consider two types of faults throughout the paper, crash and Byzantine faults, which we model as being controlled by a worst-case adversary. The adversary can fail processors at any time it wants as long as no more than $f$ processors are faulty at any time. When the adversary crash a processor $p$ at round $t$, it can choose which of the messages $p$ sent in round $t$ arrive to their destination. Subsequently, $p$ stops participating in the algorithm completely. When the adversary make a processor Byzantine, it can control everything it does, including the content of its messages and which messages it sends (or drops).
We define the Agreement problem as follows.
\begin{definition}[Agreement]
    Every processor $p$ has an input value $\inp^p \in V$, for some set of possible values $V$, and must return an output value $\outp^p\in V$ that satisfies the following requirements.
\begin{itemize}
\item 
\textbf{Safety:} If a non-faulty processor outputs $v$, then every honest processor must output $v$.
\item \textbf{Validity:} If all the non-faulty processors have the same input $v$, then every non-faulty processor must output $v$.
\item 
\textbf{Termination:} Every non-faulty processor must output a value after a finite number of rounds.
\end{itemize}
\end{definition}

We note that there are two other types of validity (one stronger and one weaker) that are of interest.
\begin{itemize}
    \item \textbf{Weak Validity:} If all the processors are non-faulty and have input $v$, then every non-faulty processor must output $v$.
    \item \textbf{Strong Validity:} The decision value of the non-faulty processors must always be an input of one of the non-faulty processors.
\end{itemize}
While the definition of agreement does not depend on the type of faulty behavior of the processors, we use the term ``Crash Agreement" when considering crash faults, and ``Byzantine Agreement" when considering Byzantine faults.

\section{Graded Byzantine Agreement}


In this section, we define the \emph{Graded Byzantine Agreement (GBA)} problem and show an unauthenticated solution that requires two rounds and works correctly when $f<n/3$.
\label{sec:gba}

\begin{definition}[GBA]
Every processor $p$ has an input value $\inp^p \in V$, for some set of possible values $V$, and must return an output value $\outp^p$ and a grade $g$, where $\outp\in V$ $g\in \binset$, that satisfy the following requirements.
\begin{itemize}
\item 
\textbf{Consistency:} If an honest processor outputs $(v,1)$, then every honest processor must output $(v,g)$, for some $g\in \binset$.
\item 
\textbf{Validity:} If all the honest processors have the same input $v$, then every honest processor must output $(v,1)$.
\item 
\textbf{Termination:} Every honest processor must output a value after a finite number of rounds.
\end{itemize}
\end{definition}

\subsection{Unauthenticated Solution for super minority corruptions}
We begin by describing the solution. 
The algorithm is comprised of two rounds.
Initially, every processor sets its output value $\outp$ to its input value $\inp$, and its grade $g$ to 0.
In the first round, every processor broadcasts a \emph{vote} message for its input value $\inp$ and receives the \emph{vote} message broadcast by every other processor. Subsequently in round two, every processor counts the number of votes it received for each value $v$ 
$\in V$, 
broadcasts a \emph{confirmation} message for 
the value $v$ that received at least $n-f$ votes, if such a value exists\footnote{Claim \ref{clm:single_confirmed_value}, stated later, proves that at most one such value may exist.},
and receives the \emph{confirmation} message broadcast by every other processor.

At the end of round 2, every processor counts the number of confirmation messages it received.
In Claim \ref{clm:single_confirmed_value} we prove that it is impossible for confirmations for two different values (even by different processors) to coexist. Hence, at most one value $v$ received confirmation messages. If there exists a value $v$ for which a processor receives a confirmation message, then the processor sets its output $\outp$ to $v$. Moreover, if the processor received $f+1$ or more confirmation messages for $v$, then it sets its grade $g$ to $1$. Finally, the processor outputs $(\outp, g)$.

Intuitively, this algorithm works since it requires $n-f$ votes to create a confirmation message. Since two different values cannot get $n-f$ votes at the same time (because $f<n/3$), only one value can be confirmed, and if an honest processor receives at least $f+1$ confirmations, it can be certain that every other honest processor receives at least one.

\begin{algorithm}
\caption{$\GBA[\inp]$ for $f<n/3$ (code for processor $p$)}
$\outp \gets \inp$, $g\gets 0$\\
\texttt{Round 1:}\\
Broadcast $\langle p, \texttt{vote}~\inp\rangle$.\\
\texttt{Round 2:}\\
\If{received $n-f$ votes for $v$} {
    Broadcast $\langle p, \texttt{confirm}~\inp\rangle$
}
\If{received at least one confirmation for $v$} 
{Set $\outp\gets v$}
\If{received at least $f+1$ confirmations for $v$} {
    Set $g \gets 1$
}
Return $(\outp, g)$
\end{algorithm}

\subsection{Correctness}
Termination holds trivially (the execution takes two rounds), and it is also not hard to verify that validity and consistency hold.
\begin{lemma}
    The algorithm satisfies validity.
\end{lemma}
\begin{proof}
If all honest processors have the same input $v$, then in round 1, every honest processor $p$ broadcasts $\langle p, \texttt{vote}~v\rangle$. Therefore every honest processor $p'$ receives at least $n-f$ votes for $v$, and  broadcasts $\langle p',\texttt{confirm}~v\rangle$ in round 2. Subsequently, every honest processor 
receives at least $n-f \geq f+1$ confirmations for $v$, so it sets $\outp\gets v$ and $g\gets 1$ and returns $(v,1)$.
\end{proof}
\begin{lemma}
    The algorithm satisfies consistency.
\end{lemma}
\begin{proof}
If there exists an honest processor that outputs $(v,1)$, then it must have seen $f+1$ confirmations for $v$. At least one of those confirmations came from an honest processor $p$, which means that every honest processor received a confirmation message from $p$ during round 2. Hence, every honest processor sets $\outp\gets v$ and outputs $(v,g)$ for some $g\in \binset$.
\end{proof}

To finish the analysis, we prove that confirmations for two different values cannot coexist.

\begin{claim}
\label{clm:single_confirmed_value}
Given that $f<n/3$.
If processor $p$ sends $\langle p, \texttt{confirm} ~v\rangle$ and processor $p'$ sends
$\langle p', \texttt{confirm} ~v'\rangle$, then $v=v'$
\end{claim}
\begin{proof}
Assume towards contradiction that $v\neq v'$. Hence, $p$ received $n-f$ votes for $v$, at least $n-2f$ of which are from honest processors, and $p'$ received $n-f$ votes for $v'$, at least $n-2f$ of which are from different honest processors. Hence, the system contains at least $2(n-2f)+f$ distinct processors, which implies that $2(n-2f)+f \leq n$. But this implies that $f\geq n/3$, contradiction.
\end{proof}

\section{Recursive crash agreement}

In this section we present algorithms for crash resilient agreement that work for any number $f<n$ of crash faults. In Subsection {\ref{sec:base_rec_alg}}, we present an algorithm that uses a recursive structure inspired by the work of Momose and Ren \cite{momose_ling}. It achieves round complexity $O(n)$ and awake complexity $O(\log n)$. Then, in Subsection {\ref{sec:optimized_rec_alg}}, we show an optimized version of the basic algorithm of Subsection {\ref{sec:base_rec_alg}}, which further improves the round and awake complexity to $O(f)$ and $O(\log f)$, respectively.

\subsection{Basic Construction}
\label{sec:base_rec_alg}

Denote by $Q$ the set of all $n$ processors.
We recursively define the two sides of a set $Q_w$ as $Q_{2w}$ and $Q_{2w+1}$ where $Q_{2w}$ is the set of the first $\ceil{|Q_w|/2}$ processors in $Q_w$, $Q_{2w+1}$ is the set of the last $\floor{|Q_w|/2}$ processors of $Q_w$, and $Q_1 = Q$.

\begin{algorithm}[ht]
\caption{$\RCA[Q_w,\inp]$ code for processor $p$}
\label{alg:rec_crash_agr}
\If{$|Q_w| \leq c$} 
{run any crash agreement algorithm and return the result.
\label{step:2}
}
\If{$p \in Q_{2w}$} {
        Run $v\gets \RCA[Q_{2w},\inp]$ \label{line:first_rec_call} . \\ 
        Send $v$ to every processor in $Q_w$\\
        Set $\outp\gets v$.
    } \Else {sleep 
    for the number of rounds it takes  $\RCA[Q_{2w},\inp]$ to terminate}
    \If{$p \in Q_{2w+1}$}  {
        Set $\inp' \gets \inp$ \\
        \If{received a value $v$ from any processor in $Q_{2w}$} {set $\inp' \gets v$}
        \label{line:after_first_rec}
Run $\outp \gets \RCA[Q_{2w+1},\inp']$.
\label{line:second_rec_call}
}

    \Else{sleep 
     for the number of rounds it takes $\RCA[Q_{2w+1},\inp]$ to terminate.}
    Return $\outp$
\end{algorithm}

\subsubsection*{Correctness}

Towards proving the correctness of the algorithm, we give a few useful definitions. We note that throughout the recursive execution, $\RCA[Q_x,y]$ is called at most once (by each processor) for every possible value of $x$ (the input $y$ might differ between processors).
For the unique execution of $\RCA[Q_x,y]$, let $F_x$ be the set of processors that crash during the execution. Conversely, let $H_x = Q_x\setminus F_x$ be the set of processors that did not crash before the end of the execution.
For every set $Q_x$, after $\RCA[Q_{2x},y]$ finishes its execution, every processor $p$ in $H_{2x}$ sends its output to every processor in $Q_{2x+1}$. We refer to this step as the dissemination step. For a processor $p\in H_{2x}$ we say that $p$ finished dissemination successfully if $p$ sends $v$ to every processor in $Q_{2x+1}$ (before crashing), and unsuccessfully otherwise.
Denote the set of processors in $H_{2x}$ that finished dissemination successfully by $HD_{2x}$.

\begin{lemma}
Algorithm $\RCA[Q_w,\inp]$ satisfies validity.
\end{lemma}
\begin{proof}
We prove the claim by induction on the size of $Q_{w}$.
For the induction basis,
if $|Q_w| \leq c$, then the algorithm runs a crash agreement and validity is guaranteed.
Now consider $|Q_w| > c$. If all non-faulty processors start with the same input $v$, then $\RCA[Q_{2w},\inp]$ returns $v$ at any processor 
in $H_{2w}$, since it satisfies validity by the induction hypothesis (since $|Q_{2w}|< |Q_{w}|$).
Thus, every non-faulty processor in $Q_{2w+1}$ runs $\RCA[Q_{2w+1},\inp']$ with input $v$ (since it either receives $v$ from a processor in 
$H_{2w}$, or it uses its own input $\inp = v$).
By the induction hypothesis, $\RCA[Q_{2w+1},\inp']$ also satisfies validity and every processor in $H_{2w+1}$ returns $v$. Hence, every non-faulty processor in $Q_w$ sets $\outp  \gets v$ and outputs $v$.
\end{proof}

\begin{lemma}
Algorithm $\RCA[Q_w,\inp]$ satisfies agreement.
\end{lemma}
\begin{proof}
We prove the claim by induction on the size of $Q_{w}$.
For the induction basis,
if $|Q_w| \leq c$, then the algorithm runs a crash agreement and agreement is guaranteed.
Now consider $|Q_w| > c$.
By the induction hypothesis, $\RCA[Q_{2w},\inp]$ satisfies agreement. Let $v$ be the decision value of $\RCA[Q_{2w},\inp]$.
We consider the following two cases.

\noindent
(a) 
$HD_{2w}\neq \emptyset$:
In this case, there exists a processor $p\in HD_{2w}$ that broadcasts $v$ successfully during the dissemination step. Hence, every processor in $Q_{2w+1}$ executes $\RCA[Q_{2w+1},\inp']$ with input $v$, and by the validity of $\RCA[Q_{2w+1},\inp']$, every 
processor in 
$H_{2w+1}$ returns $v$ from $\RCA[Q_{2w+1},\inp']$. Hence, every non-faulty processor in $Q_w $ sets $\outp\gets v$ and returns $v$ from the execution of $\RCA[Q_w,\inp]$.

\noindent
(b) 
$HD_{2w} = \emptyset$:
In this case, every processor in $H_{2w}$ crashed before the start of $\RCA[Q_{2w+1},\inp']$, therefore $H_w \subseteq H_{2w+1}$. By the induction hypothesis, $\RCA[Q_{2w+1},\inp']$ guarantees agreement. Let $v$ be the decision value of $\RCA[Q_{2w+1},\inp']$. Hence, every processor in $H_w \subseteq H_{2w+1}$ sets $\outp\gets v$ and returns $v$ from the execution of $\RCA[Q_w,\inp]$. Thus, agreement is satisfied by the execution of $\RCA[Q_w,\inp]$. 
\end{proof}

\paragraph{Complexity.}
The time complexity of Algorithm \ref{alg:rec_crash_agr} is $O(|Q|) = O(n)$, since the Time complexity follows the formula $T(n) = O(1) + 2T(n/2)$.
where $T(x) = O(1)$ for every $0<x\leq c$.
The awake complexity is $O(\log n)$ since it follows the formula $A(n) = O(1) + A(n/2)$
where $A(x) = O(1)$ for every $0<x\leq c$.

\paragraph{Note on lower bounds.} As a first step towards a lower bound on the awake complexity of agreement, one may consider subclasses of agreement algorithms (e.g., the subclass $\cA_\pi$ of algorithms that satisfy some additional property $\pi$). One property that sometimes comes up in the context of binary agreement, i.e, $V=\binset$ is \emph{$1-$preference}, i.e, the requirement that every processor that receives a 1-valued message must output 1.
We remark that while Algorithm \ref{alg:rec_crash_agr} does not satisfy the 1-preference property, it can be modified slightly so as to satisfy 1-preference with the same (asymptotic) complexities. This implies that at least the subclass $\cA_{1-preference}$ of algorithms 
does not admit a lower bound higher than $\log n$ on its worst-case awake complexity. The required modification is as follows. First, replace the crash agreement algorithm of Step \ref{step:2} with one satisfying the 1-preference property.
Second, add a preliminary round of communication at the beginning of each recursive call, where every processor $p \in Q_{2w+1}$ sends its input to every processor $p' \in Q_{2w}$. Every $p' \in Q_{2w}$ that receives a 1 (from at least one processor in $Q_{2w+1}$) overwrites its own input to 1. 
The proof that this modification transforms Algorithm \ref{alg:rec_crash_agr} into one satisfying 1-preference is deferred to appendix \ref{sec:deferred_proofs}.

\subsection{Optimized Construction}
\label{sec:optimized_rec_alg}
Here we present an optimization that uses Algorithm {\ref{alg:rec_crash_agr}} as building block, and improves its round and awake complexity from $O(n)$ and $O(\log n)$ to $O(f)$ and $O(\log f)$ respectively.

The idea behind the algorithm is to split the set of processors $Q$ into disjoint subsets of size $f+1$ $Q^1, Q^2, \ldots, Q^s$ for $s = \floor{n/(f+1)}$ (with an additional smaller subset $Q^{s+1}$ for the remaining processors if needed). Then, run an instance of Algorithm {\ref{alg:rec_crash_agr}} on each subset (excluding $Q^{s+1}$) in parallel ($\RCA[Q^i]$ for $i\in[1,s]$). After each instance terminates, and every subset $Q^i$ agrees on a decision value $v_i$, every processor broadcasts the decision value of the subset in belongs to. Since $|Q^i| =f+1$, at least one processor in each $Q^i$ will successfully broadcast $v_i$. Lastly, every processor computes $v = \max\{v_i\}_{i\in[1,s]}$ and outputs $v$.

\begin{algorithm}
    \caption{Optimized $\RCA$}
    \label{alg:optimized_rec_crash_agr}
    Split $Q$ into $s=\floor{n/(f+1)}$ subsets of size $f+1$ each $Q^1, Q^2, \ldots, Q^s$ and let $Q^{s+1}$ be the set of the remaining $n-(f+1)s$ processors.\\
    In parallel, execute $\RCA[Q^i]$ and let $v_i$ be the returned value for every $i\in [1,s]$.\label{line:multi_rec_executions}\\
    For $i\in [1,s]$, if $p \in Q^i$, send $v_i$ to every processor \label{line:broadcast_sub_decisions}. 
    \\
    Collect the subset decision values $V=(v_1, v_2, \ldots,v_s)$ (use $\bot$ for any missing value), and set $\outp \gets \max \{v_i\}_{i\in[1,s]}$ \label{line:collect_sub_decisions}.
    \\
    Return $\outp$.
    
\end{algorithm}
\subsubsection*{Correctness}
\begin{lemma}
    Algorithm \ref{alg:optimized_rec_crash_agr} satisfies validity
\end{lemma}
\begin{proof}
If every non-faulty processor has the same input $v$, then by the validity of Algorithm \ref{alg:rec_crash_agr}, $v_i = v$ for every $i\in [1,s]$. Hence, in Line \ref{line:collect_sub_decisions}, every processor 
sets $\outp = \max \{v\}_{i\in[1,s]}=v$. Subsequently, every processor in $Q^{s+1}$ receives $v$ and sets $\outp\gets v$. Thus $\outp = v$ for every non-faulty processor. 
\end{proof}



\begin{lemma}
    Algorithm \ref{alg:optimized_rec_crash_agr} satisfies agreement
\end{lemma}
\begin{proof}
Since Algorithm \ref{alg:rec_crash_agr} satisfies agreement, every 
processor in $Q^i$ that did not crash during the execution of $\RCA[Q^i]$ returns the same value $v_i$ from the execution of $\RCA[Q^i]$ in Line \ref{line:multi_rec_executions}, for every $i\in [1,s]$. Also, since $|Q^i|=f+1$, at least one processor in $Q^i$ will successfully broadcast $v_i$ in Line \ref{line:broadcast_sub_decisions}, for every $i\in [1,s]$. Subsequently, every non-faulty processor successfully constructs $V = (v_1,v_2, \ldots, v_s)$ ($v_i \neq \bot$ for every $i\in [1,s]$) and sets $\outp \gets v$ for $v = \max\{v_i\}_{i\in[1,s]}$. Therefore, every non-faulty outputs the same value, satisfying agreement.
\end{proof}
\paragraph{Complexity}
The time complexity of each recursive call in step $2$ is $O(f)$. Since they are performed in parallel, the total time complexity of step 2 is $O(f)$. There are two more rounds of communication in steps 3 and 4, thus the total number of rounds is $O(f)$.

Since the subsets are pairwise disjoint and the awake complexity of each recursive call in step 2 is $O(\log f)$, the awake complexity of step 2 is $O(\log f)$. After that there are at most two additional awake rounds bringing the overall awake complexity of the algorithm to $O(\log f)$.

\bibliographystyle{plainurl}
\bibliography{refs}

\clearpage
\centerline{\Large\bf Appendix}
\appendix
\section{Deferred proofs}
\label{sec:deferred_proofs}


\begin{lemma}
The modified Algorithm \ref{alg:rec_crash_agr} satisfies 1-preference.
\end{lemma}
\begin{proof}
Assume that there exists a processor $p$ that received 1 during the execution of $\RCA[Q_w, y]$, and did not crash before the end of the execution ($p \in H_{w}$).
We prove the Lemma by induction on the size of $Q_w$. For the induction basis, if $|Q_w| \leq c$, then the crash agreement algorithm invoked in Step \ref{step:2} of the algorithm
satisfies 1-preference, hence $p$ outputs 1.
Now suppose $|Q_w| > c$. We consider two cases. 
\\
(1) $p \in Q_{2w}$: Then $p$ receives all of its messages before or during the execution of $\RCA[Q_{2w}, \inp]$ in Line \ref{line:first_rec_call}. Since $\RCA[Q_{2w}, \inp]$ satisfies 1-preference by the induction hypothesis, every processor $p'\in H_{2w}$ returns 1. Hence, $p$ decides $1$.
\\
(2) $p \in Q_{2w+1}$: Then $p$ either has input 1, or it receives 1 after the first recursive call in Line \ref{line:after_first_rec}.
\\
$\bullet$ If the input of $p$ is 1, then it successfully sends $1$ to every processor in $Q_{2w}$, before the execution of the first recursive call in Line \ref{line:first_rec_call}. Hence, $\RCA[Q_{2w}, \inp]$ returns $1$ at every processor $p\in H_{2w}$. Therefore, $p$ runs the second recursive call $\RCA[Q_{2w+1}, \inp]$ in Line \ref{line:second_rec_call} with $\inp = 1$ (because it either uses its own input value which is 1, or it uses the decision value of the first recursive call which is also 1).
By the induction hypothesis, $\RCA[Q_{2w+1}, \inp]$ satisfies 1-preference, hence it returns 1 at every $p' \in H_{2w+1}$. Therefore, $p$ decides 1.
\\
$\bullet$ If $p$ receives 1 after the first recursive call, then it runs $\RCA[Q_{2w+1}, \inp]$ with $\inp = 1$ and again by the induction hypothesis, $\RCA[Q_{2w+1}, 1]$ returns 1 at every processor $p' \in H_{2w+1}$. Therefore, $p$ decides 1.
     \end{proof}
\end{document}